\documentclass[a4paper]{article}
\usepackage{amsmath}
\usepackage[english]{babel}
\usepackage{latexsym}
\usepackage{amssymb}
\usepackage{amscd}
\usepackage{amsgen,amstext,amsbsy,amsopn}
\usepackage{math rsfs}

\usepackage{amsthm,epsfig,graphicx,graphics}
\usepackage[latin1]{inputenc}
\usepackage{xspace}
\usepackage{amsxtra}
\usepackage{bm,bbm}

\usepackage{color}

\newcommand{\version}{\today}
\numberwithin{equation}{section}
\newcommand{\bdm}{\begin{displaymath}}
\newcommand{\edm}{\end{displaymath}}
\newcommand{\bdn}{\begin{eqnarray}}
\newcommand{\edn}{\end{eqnarray}}
\newcommand{\bay}{\begin{array}{c}}
\newcommand{\eay}{\end{array}}
\newcommand{\ben}{\begin{enumerate}}
\newcommand{\een}{\end{enumerate}}
\newcommand{\beq}{\begin{equation}}
\newcommand{\eeq}{\end{equation}}

\newcommand{\tx}{\textstyle}

\newcommand{\lf}{\left}
\newcommand{\ri}{\right}

\newcommand{\ket}[1]{\lf|#1 \ri\rangle}

\newcommand{\expval}[3]{\lf\langle #1 \ri|\lf.  #3 \ri. \lf|  #2 \ri\rangle}

\newcommand{\kett}[1]{\lf|#1 \ri)}

\newcommand{\expvalt}[3]{\lf( #1\lf| \lf. #3 \ri. \ri| #2 \ri)}

\newcommand{\xv}{\mathbf{x}}

\newcommand{\yv}{\mathbf{y}}

\newcommand{\kv}{\mathbf{k}}

\newcommand{\diff}{\mathrm{d}}
\newcommand{\bett}{\tilde{\beta}}
\newcommand{\eps}{\varepsilon}

\newcommand{\p}{\partial}

\newcommand{\tr}{\mathrm{Tr}}

\newcommand{\hamb}{\mathcal{H}}

\newcommand{\hambd}{\mathcal{H}^{D}}
\newcommand{\kb}{\mathcal{K}}

\newcommand{\latt}{\Lambda}
\newcommand{\spin}{\hat{S}}

\newcommand{\spinv}{\hat{\mathbf{S}}}
\newcommand{\tspinv}{\hat{\mathbf{S}}_T}

\newcommand{\fspinv}{\hat{\bm{\mathcal{S}}}}

\newcommand{\up}{a^{\dagger}}
\newcommand{\ax}{a_{\xv}}

\newcommand{\ak}{\tilde{a}_{\kv}}
\newcommand{\upx}{\up_{\xv}}

\newcommand{\upk}{\tilde{a}^{\dagger}_{\kv}}

\newcommand{\hn}{\hat{n}}
\newcommand{\nx}{n_{\xv}}

\newcommand{\hnx}{\hat{n}_{\xv}}

\newcommand{\tnk}{\tilde{n}_{\kv}}
\newcommand{\no}{\tilde{n}_{0}}

\newcommand{\hilb}{\mathscr{F}}

\newcommand{\Z}{\mathbb{Z}}

\newcommand{\HH}{\mathscr{H}}

\newcommand{\OO}{\mathcal{O}}

\newcommand{\ZZ}{\mathcal{Z}}

\newcommand{\one}{\mathbbm{1}}

\def\b{\beta}
\def\p{\pi}
\def\d{\delta}
\def\L{\Lambda}
\def\e{\varepsilon}

\def\kk{{\bf k}}
\def\xx{{\bf x}}
\def\yy{{\bf y}}

\def\be{\begin{equation}}
\def\ee{\end{equation}}
\def\bea{\begin{eqnarray}}
\def\eea{\end{eqnarray}}

\def\media#1{{\left\langle#1\right\rangle}}

\def\dpr{\partial}

\newtheorem{teo}{Theorem}[section]

\newtheorem{pro}{Proposition}[section]

\newcounter{remark}[section]

\newenvironment{rem}{\stepcounter{remark} \vspace{0,1cm} \noindent \textit{Remark \thesection.\theremark}\,}{\vspace{0,2cm}}

\begin{document}

\markboth{\scriptsize{Free Energy of the Heisenberg Model -- \textsc{Correggi, Giuliani} -- \version}}{\scriptsize{Free Energy of the Heisenberg Model -- \textsc{Correggi, Giuliani}  -- \version}}

\title{The Free Energy of the Quantum Heisenberg Ferromagnet \\ at Large Spin}
\author{M. Correggi, A. Giuliani	\\
	\mbox{}	\\
	\normalsize\it Dipartimento di Matematica, Universit\`{a} degli Studi Roma Tre,	\\
	\normalsize\it L.go S. Leonardo Murialdo 1, 00146, Rome, Italy.}	
\date{\version}

\maketitle
\abstract We consider the spin-$S$ ferromagnetic Heisenberg model in three dimensions, in the absence of 
an external field. Spin wave theory suggests that in a suitable temperature regime the system 
behaves effectively as a system of non-interacting bosons (magnons). 
We prove this fact at the level of the specific free energy: if $S\to\infty$ and the inverse temperature $\b\to 0$ in such 
a way that 
$\b S$ stays constant, we rigorously show that the free energy per unit volume 
converges to the one suggested by spin wave theory. 
The proof 
is based on the localization of the system in small boxes and on upper and lower bounds on the 
local free energy, and it also provides explicit error bounds on the remainder. 

\section{Introduction}

An important open problem in theoretical and mathematical physics is the proof 
of long range order in the three-dimensional (3D) quantum Heisenberg ferromagnet (FM) at low temperatures. 
While the existence of long range order at low temperatures in the 
classical Heisenberg model and in the quantum Heisenberg antiferromagnet can be proved by 
reflection positivity methods \cite{FSS,DLS}, the broken phase of the 
quantum ferromagnet eluded any rigorous treatment so far. 

From a heuristic point of view, a very useful and suggestive representation of the quantum FM 
is in terms of spin waves, an idea first introduced by Bloch in his seminal work \cite{B1,B2}. The spin waves are 
the lowest energy excitations, which give the dominant contribution to the free energy at low temperatures; they 
satisfy a Bose statistics and are in many respects the analogues of the phonons in crystals (see, e.g., \cite{K} for a 
classical and comprehensive review). 
Bloch's theory was later generalized in several directions by Herring and Kittel \cite{HK}, Holstein and Primakoff 
\cite{HP}, Dyson \cite{D}, and it was used, among other things, to compute a low temperature expansion for the 
spontaneous magnetization in zero external magnetic field: after a few erroneous attempts \cite{Kr, Op, Sch, VK}, 
Dyson's result \cite{D2} was confirmed by a number of different methods \cite{CH,KL,M,O,RL,Sz,VLP,W,YW,Z} and 
further extended, 
more recently, by the effective Lagrangian method \cite{Ho1,Ho2}. The conclusion is that at low energies the 
corrections 
to the simple Bloch's theory coming from the interactions among spin waves are so small that for most practical 
purposes the linear theory is enough, both in the presence or in the absence of an external magnetic field. 

While physically Bloch's theory is accepted and in good agreement with experiments, from a more mathematical 
point of view there is no confirmation of its correctness yet. It is fair to say that the current mathematical methods 
of quantum many body systems are still far from allowing us to prove the existence of a spontaneous magnetization
in the quantum Heisenberg ferromagnet and to possibly  confirm the exactness of Dyson's computation \cite{D2}.

An easier, but still unsolved problem is to prove the correctness of spin wave theory at the level of thermodynamic 
rather than correlation functions, at least in the large spin limit, where the quantum FM is (formally) close to 
its classical counterpart \cite{L} or to the free Bose gas \cite{HP,vHBW}, depending on the temperature regime 
which we look the system at. Some progress in this direction is due to Conlon and Solovej \cite{CS1}, who 
proved that the free energy of the quantum Heisenberg FM {\it in the presence of a large magnetic field} is the same 
as the one of 
a free Bose gas, asymptotically for large on-site spin $S$. For small (i.e., $1/2$) spin and small temperatures,  
Conlon and Solovej \cite{CS2} and Toth \cite{T} derived an upper bound on the free energy, close but not exactly 
equal to Bloch's theory prediction. 
In the large spin limit, it was unclear whether the rigorous justification of spin wave theory by Conlon and Solovej 
\cite{CS1} could be extended to the case of zero magnetic field, which is much more subtle and unstable with respect 
to temperature-induced excitations, due to the global rotational symmetry of the model. 
For the classical counterpart of this problem, that is the convergence of the large-S classical Heisenberg model 
to the gaussian one, Conlon and Solovej managed to show the uniformity of their results as the magnetic field is 
taken to zero \cite{CS3} (an even stronger convergence result to the gaussian model for the case of the classical rotator model is due to Bricmont et al \cite{BFLLS}). 
However, their proof did not apply to the quantum case, see the comments after \cite[Theorem 2.1]{CS3}. The large spin limit for the XXZ chain in the absence of a magnetic spin was studied by 
Michoel and Nachtergaele
\cite{MN1,MN2}, but their proof required the presence of an energy gap in the spectrum of the spin chain. 

In this paper, we solve the issue by extending these results to the 3D quantum rotationally invariant (gapless) case;
namely, we show that the free energy of the 3D quantum Heisenberg FM at temperatures of order $S$ 
converges as $S\to\infty$ to the one of the non-interacting magnon gas. This rigorously proves (in a weak sense)
the validity of the spin wave approximation at leading order in $S$ for $S$ large. 
We use the Holstein-Primakoff representation of the quantum Heisenberg FM in terms of interacting bosons, and we 
take advantage of the methods and the ideas recently developed by Lieb, Seiringer and Yngvason to treat 
the interacting Bose gas at low densities \cite{LS,LSY,LY} (see \cite{LSSY} for a review): key ingredients of our proof 
are: (i) a localization bound, allowing us to coarse grain the system in finite boxes; (ii) a priori bounds on the local 
energy, which make use of the energy gap in the box and allow us to drop the states with total spin far from the 
maximum. 

The paper is organized as follows. In Section \ref{model: sec} we define the model and state our main results. In 
Section \ref{bosons: sec} we review the Holstein-Primakoff representation and briefly review spin wave theory. In 
Section \ref{proofs: sec} we prove our main theorem by deriving suitable lower (Section \ref{proofs lb: sec}) and 
upper (Section \ref{proofs ub: sec}) bounds. 

\subsection{The Model and the Main Result}
\label{model: sec}

We consider a ferromagnetic Heisenberg model with nearest neighbor interactions which is associated with the 
Hamiltonian\footnote{We use the convention of distinguishing (when needed) between operators and numbers or 
eigenvalues by means of a $ \hat{} $ on top of the first ones. As usual we also denote vectors by bold letters.}
\beq
	\label{heisenberg ham 1}
	H_\L^{per} : = J \sum_{\media{\xv, \yv}\subset \Lambda}(S^2- \spinv_{\xv} \cdot \spinv_{\yv}),
\eeq
where $ J > 0 $, $\L$ is a cubic portion of side length $L$ of $\mathbb{Z}^3$ with periodic boundary conditions,
the sum runs over the nearest neighbor pairs in $\L$, and $ \spinv_{\xv} $ is a three-components 
spin $S$ operator. This means that the components of $\hat{\bf S}_\xx$, denoted by $\hat S^j_\xx$,
$j=1,2,3$, satisfy:
\be [\hat S^j_\xx,\hat S_\yy^k]=i\e_{jkl}\hat S^l_\xx\d_{\xx,\yy}\;, \qquad \hat{\bf S}_\xx^2=
(\hat S^1_\xx)^2+(\hat S^2_\xx)^2+(\hat S^3_\xx)^2=S(S+1)\;,\ee
where in the first equation $\e_{klm}$ is the completely antisymmetric symbol,
while in the second equation $S$ is the size of the spin, with $2S$ an integer. We shall denote 
by $ \HH_\L $ the Hilbert space of spin configurations in $\L$ such that $\hat{\bf S}_\xx^2=S_\xx(S_\xx+1)$
with $S_\xx=S$; a convenient basis for this space is  $\ket{\lf\{ S_{\xv}^3 \ri\}} : = \bigotimes_{\xv \in \latt} \ket{S_{\xv}^3}$, with $-S\le S^3_\xx\le S$, $\forall \xx\in \L$. Note that the Hamiltonian \eqref{heisenberg ham 1} is normalized in such a way that the ground state energy is zero.

Our main object of interest is the free-energy per site
\beq
	\label{free energy}
	f(S,\beta,\latt) : = - \frac{1}{\beta|\latt|} \log \ZZ(S,\beta,\latt)\;,\quad{\rm with}\quad 
	 \ZZ(S,\beta,\latt)= \tr_{\HH_{\L}} \lf( \exp \lf\{ - \beta H_\L^{per} \ri\} \ri),
\eeq
and the thermodynamic limit of the specific free energy, namely, $ f(S,\beta) : = \lim_{\latt \to \infty} f(S,\beta,\latt) $.
Our goal is to compute $f(S,\beta)$ to the leading order in $S$ for $S$ large in a suitable temperature regime (to be discussed below) and to prove that 
it coincides within explicitly estimated errors with the free energy of a free Bose gas. More precisely, our main result,
which generalizes a similar result in \cite[Theorem 4.2]{CS1} proven for a non-zero (in fact huge, i.e., $ h = \OO(S) $) magnetic field, 
can be stated as follows. 

	\begin{teo}[Free energy asymptotics]
		\label{free energy: teo}
		\mbox{}	\\
		Assume that $ \beta = \bett S^{-1}  $ as $ S \to \infty $ for some $ \bett > 0 $ constant
		independent of $S$.	Then one has
		\beq
			\label{free energy asympt}
			\frac{f(S,\beta)}{S} = \frac{1}{\tilde\b} \int\limits_{\mathcal{B}}\frac{\diff\kk}{(2\p)^3}\log \lf( 1 -e^{ - \bett J \eps(\kv)}\ri) + \OO\lf(S^{-1/8}(\log S)^{1/4}\ri),
		\eeq					
		where $\eps(\kv) : =  \sum_{i=1}^3(1-\cos k_i)$ is the dispersion relation and 
		$\mathcal{B}=[-\p,\p]^3$ is the first Brillouin zone.
	\end{teo}

	\begin{rem}(Zero temperature limit)
		\\
		By a direct inspection of the proof, it is easy to verify that all the estimates are uniform in $ \bett > 0 $, if $\bett$ is well separated from $0$. Hence the result can be straightforwardly extended to any $ \bett=\bett(S)$ function of $ S $, provided that $\liminf_{S\to\infty}\bett(S)>0$, including the case that $\lim_{S\to\infty} \bett= \infty$. This last case for the spin-$1/2$ Heisenberg ferromagnet was considered in \cite{CS2} and \cite{T}, where only upper bounds on the free energy were proven.
	\end{rem}

	\begin{rem}(External magnetic field)
		\\
		The addition of an external magnetic field would not affect the result proven in the Theorem above. In fact, as already mentioned, the case of a magnetic field $ h $ of order $ \OO(S) $ was already considered in \cite{CS1}. However any magnetic field $ h \ll S $ would produce a contribution to the r.h.s. of \eqref{free energy asympt} of order $ h/S $ and thus would not affect the leading order term, irrespective of the sign of $ h $.
	\end{rem}

The intuition behind this result is based on
a well-known bosonic representation of the Heisenberg model, first proposed by Holstein and Primakoff in \cite{HP}, reviewed 
in the following section.

\subsection{Bose Gas Representation and Magnon Approximation}
\label{bosons: sec}

It is well known since the pioneering work Holstain-Primakoff \cite{HP} that the Heisenberg Hamiltonian 
can be rewritten in terms of suitable creation and annihilation operators, so that the excitations of the model can be described as Bose particles exactly as phonons in crystals. For any $ \xv \in \latt $ we set
\beq
	\label{ax upx}
	\spin_{\xv}^+ = : \sqrt{2S} \upx \sqrt{1 - \frac{\upx\ax}{2S}},		\qquad	\spin_{\xv}^{-} = : \sqrt{2S} \sqrt{1 - \frac{\upx\ax}{2S}} \ax,	\qquad	\spin_{\xv}^z = : \upx\ax - S,
\eeq
where $a^{\dagger}_{\xv}, a_\xv$ are bosonic creation and annihilation operators.  We can associate 
with such Bose modes an Hilbert space $ \mathscr{F}_\L$ isomorphic to the spin Hilbert space $ \HH_\L $ by 
diagonalizing the number of particles operator $\hnx : = \upx \ax$ and we denote such a basis by
$\kett{\{ \nx \}} : =  \bigotimes_{\xv \in \latt} \kett{ \nx}$.
Note that the constraint $ - S \leq S_{\xv}^z \leq S $ translates into the requirement 
\beq	\nx \leq 2S,\eeq
i.e., the Hilbert space is truncated and the occupation number per site can not exceed $ 2 S$. The 
isomorphism between the Hilbert spaces $ \HH_\L$ and $ \mathscr{F}_\L$ can be implemented by means of the 
one-to-one correspondence
$\kett{\nx} \longleftrightarrow \ket{S_{\xv}^3 = \nx-S}$,
so that, for instance, the ground state with all the spins pointing down is $ \kett{ \{ \nx = 0 \}} $.

The Hamiltonian $ H_\L^{per}$ can be rewritten in terms of the creation and annihilation operators as
\beq
 	\label{hambh}
	\mathcal{H}_\L^{per} =  JS\sum_{\substack{\xv, \yv \in \Lambda\\ |\!|\xv - \yv|\!|_\L= 1}} 
	\bigg\{ - \upx \bigg( 1 - \frac{\upx\ax}{2S} \bigg)^{1/2} \bigg( 1 - \frac{\up_{\yv}a_{\yv}}{2S} 
	\bigg)^{1/2} a_{\yv} +   \upx a_{\xv} -\frac1{2S} \upx \up_{\yv} \ax a_{\yv}\bigg\},
\eeq
where $|\!|\cdot|\!|$ is the euclidean distance on the torus $\L$. 
For $S$ large, 
this representation naturally leads to a decomposition of $\mathcal{H}_\L^{per}$ into ``free" and ``interacting" 
parts: 
\beq
	\label{hamob}
	\hamb_\L^{per} = \hamb_{0,\L}^{per} + \kb_\L^{per}, 	\qquad \hamb_{0,\L}^{per} : =   SJ \sum_{\substack{\xv, \yv \in \Lambda\\ 
	|\!|\xv - \yv|\!|_\L = 1}} \upx \lf( \ax - a_{\yv} \ri),
\eeq
where $ \kb_\L^{per}$ is at least quartic in the Bose operators and is formally of relative size $1/S$ with respect to the quadratic part. 
The magnon approximation simply amounts to drop the interaction term $ \kb_\L^{per}$ from the energy $ \hamb_\L^{per}$ and study the free Bose gas so obtained. By means of the Fourier transform the free Hamiltonian $ \hamb_{0,\L}^{per}$  can be easily diagonalized: given any $ \kv \in \latt^* $, where $ \latt^* $ denotes the dual lattice, i.e., 
	$
		\label{modes}
		\kv = \tx\frac{2\pi}{L} {\bf m}	
	$, with $ 0 \leq m_i < L $,
	we set $\eps(\kv) : =  \sum_{i=1}^3(1-\cos k_i)$
	and define the creation and annihilation operators for the Fourier modes as
	$\ak : = L^{-3/2} \sum_{\xv \in \latt} \exp \lf\{ i \kv \cdot \xv \ri\} \ax	$,
	so that the Hamiltonian \eqref{hamob} can be rewritten 
	$\hamb_{0,\L}^{per} = SJ \sum_{\kv \in \latt^*} \eps(\kv) \upk \ak $.
	
	The free energy of such a Bose gas should be computed by taking the trace of 
$ \exp \{ - \beta \hamb_{0,\L}^{per} \} $ on the truncated Hilbert space $ \mathscr{F}_\L$, where the occupation number of 
each site cannot exceed $ 2S $. However, in the large $ S $ limit one expects that such a constraint 
can be removed in the calculation of the free-energy, up to higher order corrections, and in the 
thermodynamic limit this would yield exactly the dominant contribution in the r.h.s. of \eqref{free 
energy asympt}.

\section{Proofs}
\label{proofs: sec}

In this section we prove Theorem \ref{free energy: teo}. The asymptotics \eqref{free energy asympt} is obtained by comparing suitable upper (Section \ref{proofs ub: sec}) and lower (Section \ref{proofs lb: sec}) bounds to the free energy of the ferromagnetic Heisenberg model. 

\subsection{Lower Bound}
\label{proofs lb: sec}

Here we prove the lower bound to the free energy. The main ingredients in the proof are a localization 
of the energy into Neumann boxes and two simple lower bounds on the localized Hamiltonian: the first one uses the 
energy gap in the box to show that the states with total spin far from the maximum have large energy;
the second one is a rough bound on the contribution from the bosonic interaction part, which will be useful 
only after having restricted the trace to states with 
third component of the total spin close to the minimum value (the proofs of these two estimates are deferred to Appendix \ref{app.tec}).

	\begin{pro}[Free energy lower bound]
		\label{free energy lb: pro}
		\mbox{}	\\
		Assume that $ \beta = \bett S^{-1}  $ as $ S \to \infty $ for some $ \bett > 0 $ constant.	Then one has
		\beq
			\label{free energy lb}
			\frac{f(S,\beta)}S \geq   \frac{1}{\tilde\b} \int\limits_{\mathcal{B}}\frac{\diff\kk}{(2\p)^3}\log \lf( 1 -e^{ - \bett J \eps(\kv)}\ri)  - \OO\lf(S^{-1/8} (\log S)^{1/4}\ri).
		\eeq					
			\end{pro}

	\begin{proof}
		The first step towards the proof of \eqref{free energy lb} is a localization of the energy into 
small boxes  of side length $ \ell $ with Neumann conditions at the boundary. We partition the big box 
$\L$ into boxes $\L_i$ of side $\ell$ and use the positivity of the bond energy for all the bonds 
connecting a site in $\L_i$ with a site in $\L_j$, $i\neq j$.  Correspondingly, we bound the original 
Hamiltonian $H$ from below as $H\ge \sum_i H^N_{\L_i}$ where $ H_{\L_i}^N $ depends only on 
the degrees of freedom associated with the spins in $ \L_i $ and has free conditions at the boundary:
\be H^N_{\L_i} =\frac{J}2\sum_{\substack{\xx,\yy\in \L_i\\ |\xx-\yy|=1}}(S^2-\hat{\bf S}_\xx\cdot\hat{\bf S}_\yy)\;,\ee
where $|\cdot|$ is the euclidean distance on $\mathbb{Z}^3$.
The Hilbert space on which $ H_{\L_i}^N $ is assume to act (i.e., the restriction of $\HH_\L$ to $\L_i$)
will be denoted by $\HH_{\L_i}$.

Obviously, all the Hamiltonians $H^N_{\L_i}$ commute among each other. Therefore, 
\be  \ZZ(S,\beta,\latt) \leq \Big[ \tr_{\HH_{\L_1}} \lf( \exp \lf\{ - \beta H^N_{\L_1} \ri\} \ri)\Big]^{L^3/\ell^3}=: \ZZ^N(S,\beta,\latt_1)^{L^3/\ell^3},\ee
where the trace $\tr_{\HH_{\L_1}}$ is only over the spin degrees of freedom within $\L_1$.  

In the computation of  $\tr_{\HH_{\L_1}} \exp\{-\b H^N_{\L_1} \}$ we now distinguish between the states with 
total spin close to the maximum, 
from those with ``small spin", which have a big energy and, thus,  give a small contribution
to the free energy. Note that the Hamiltonian $ H_{\L_1}^N $ is invariant under global rotations, that 
is $H_{\L_1}^N$ commutes with the three components of the total spin $\hat{\bf S}_T:=\sum_{\xx\in\L_1}\hat{\bf S}_\xx$ in $\L_1$. Therefore, $ H_{\L_1}^N $ is block diagonal with respect to the 
decomposition $\HH_{\L_1}=\bigoplus_{S_T=0}^{S\ell^3}\bigoplus_{S_T^3=-S_T}^{S_T}\HH_{S_T,S_T^3}$, where $\HH_{S_T,S_T^3}$
is the subspace of $\HH_{\L_1}$ on which $\hat{\bf S}_T^2=S_T(S_T+1)$ and $\hat S^3_T=S^3_T$. On each $\HH_{S_T,S^3_T}$, 
the Hamiltonian can be bounded from below as follows, independently of the value of $S^3_T$ (see Appendix \ref{app.tec} for a proof).

	\begin{pro}[Lower bound on $ H_{\L_1}^N $]
		\label{H lower bound: pro}
		\mbox{}	\\
		There exists a positive constant  $ c > 0 $ such that, if 
		$S - \ell^{-3}S_{T} > 3 J c^{-1} \ell^2$, then 
		\beq
			\label{H lower bound}
			{H_{\L_1}^N} \Big|_{\HH_{S_T,S_T^3}}\geq c \ell S \lf( S - \frac{S_{T}}{\ell^3}
			 \ri).
		\eeq
	\end{pro}
\vskip.2truecm
We now split the trace of interest in two parts:
\beq	\label{step 2}
 \ZZ^N(\beta,S,\Lambda_1) = \sum_{S_{T}\geq S_{\star} \ell^3} \sum_{S_T^3 = - S_T}^{S_T}  {\rm Tr}_{\HH_{S_T,S_T^3}} \exp \lf\{ - \beta H^{N}_{\L_1} \ri\} + \mathcal{R},
		\eeq
		where $ S_{\star} < S $ is a parameter which is going to be fixed later, and 
		the rest $\mathcal{R}$ 
		can be bounded by means of \eqref{H lower bound} as follows:
		\beq
 			\label{rest 1}
			\mathcal{R} = \sum_{S_{T}< S_{\star} \ell^3}  \sum_{S_T^3 = - S_T}^{S_T} {\rm Tr}_{\HH_{S_T,S_T^3}} \exp \lf\{ - \beta H^{N}_{\L_1} \ri\} \le (2S +1)^{\ell^3} e^{ - c \bett \ell S \lf( 1 - \frac{S_{\star}}{S} \ri)},
		\eeq
		provided that $S-S_*>3Jc^{-1}\ell^2$,
		which can be satisfied by picking
\be \label{Sstar choice}
			\frac{S_{\star}}{S} = 1 - \frac{2 \ell^2 \log(2S+1)}{c \tilde\b S}
			\ee
and $\ell,S$ large enough, so that $\mathcal{R}\le (2S+1)^{-\ell^3} \ll\OO(1)$.

We are now left with the main contribution to the trace (the one involving $S_{T}\ge S_*\ell^3$), to be called 
$\tilde{\ZZ}(\beta,S,\Lambda_1)$. 
  By using once again the fact that $H^N_{\L_1}$ commutes with $\hat{\bf S}_T$, we find that 
${\rm Tr}_{\HH_{S_T,S_T^3}} \exp \big\{- \beta H^N_{\L_1} \big\}$ is {\it independent}\footnote{
This can be proved as follows. Let us indicate by $\xx_1,\ldots, \xx_{\ell^3}$ the sites of $\L_1$ 
labeled in lexicographic order. By the theory of the composition of angular momenta, 
a bona fide 
basis for $\HH_{\L_1}$ is provided by the common eigenvectors of $(\hat{\bf S}_{\xx_1}+\hat{\bf S}_{\xx_2})^2$, $(\hat{\bf S}_{\xx_1}+\hat{\bf S}_{\xx_2}+\hat{\bf S}_{\xx_3})^2$, $\ldots$, 
$(\hat{\bf S}_{\xx_1}+\cdots+\hat{\bf S}_{\xx_{\ell^3-1}})^2$, $\hat{\bf S}_T^2$, $\hat S^3_T$. In other
words, the eigenvalues of $(\hat{\bf S}_{\xx_1}+\hat{\bf S}_{\xx_2})^2$, $(\hat{\bf S}_{\xx_1}+\hat{\bf S}_{\xx_2}+\hat{\bf S}_{\xx_3})^2$, $\ldots$, 
$(\hat{\bf S}_{\xx_1}+\cdots+\hat{\bf S}_{\xx_{\ell^3-1}})^2$ can be used as good quantum numbers 
for classifying the states of $\HH_{S_T,S_T^3}$. Note that the operators associated with these quantum numbers are all scalars, i.e., they commute with the three components of $\hat {\bf S}_T$:
therefore, the eigenvectors of $H_{\L_1}^N$ on $\HH_{S_T,S_T^3}$ are invariant under the action 
of $\hat{\bf S}_T$, which implies in particular that ${\rm Tr}_{\HH_{S_T,S_T^3}} \exp\{- \beta H^N_{\L_1} \}$ is independent of $S^3_T$.} of $S^3_T$, so that 
\be \label{step 3}\tilde{\ZZ}(\beta,S,\Lambda_1)=
\sum_{S_{T} = S_{\star}\ell^3}^{S\ell^3}(2S_T+1)
{\rm Tr}_{\HH_{S_T,-S_T}}\exp \lf\{ - \beta H^{N}_{\L_1} \ri\}. \ee
We can now apply the boson representation given in Eq.\eqref{ax upx}, which implies
\be \label{step 4}\tilde{\ZZ}(\beta,S,\Lambda_1) \leq   \lf(2S\ell^3 + 1\ri)  \sum_{\substack{\{ n_{\xv} \}, \xv \in \L_1 \:\\
 \sum \nx \leq \lf(S - S_{\star}\ri)\ell^3}} \expvalt{\{ n_{\xv} \}}{\{ n_{\xv} \}}{\exp\lf\{ - \beta \hamb^N_{\L_1} \ri\}},
		\ee
where $\hamb^N_{\L_1}$ is the bosonic hamiltonian in $\L_1$ with Neumann boundary conditions 
(i.e., it is given in Eq. \eqref{hambh} with $\L$ replaced by $\L_1$ and the condition $|\!|\xx-\yy|\!|_\L=1$ replaced by $|\xx-\yy|=1$). Inspired by 
Eqs. \eqref{hambh}--\eqref{hamob}, we rewrite $\hamb^N_{\L_1}=\hamb^N_{0,\L_1}+ \kb^N_{\L_1} $,
with $\hamb^N_{0,\L_1}=JS\sum_{\media{\xx,\yy}\subset \L_1}(a^\dagger_\xx-a^\dagger_\yy)
(a_\xx-a_\yy)$ and $ \kb^N_{\L_1} $ the interaction part, which can be bounded as follows (see Appendix \ref{app.tec} for a proof).
		
\begin{pro}[Estimate of $ \kb^N_{\L_1} $]
\label{est kb: pro}
\mbox{}	\\
There exists a finite constant $ C $ such that
\beq 	\label{kb est}\lf|  \kb^{N}_{\L_1} \ri| \leq  C \hat{N}_{\L_1}^2,\eeq	
where $ \hat{N}_{\L_1} : = \sum_{\xv \in \latt_1} \hnx $.
\end{pro}

Using this estimate in (\ref{step 4}) together with the fact that 
$\hat N_{\L_1}\le (S-S_*) \ell^3=(2/c\tilde\b)\ell^5\log(2S+1)$, we get for a suitable constant $C'$
\beq	\label{step 5}
\tilde{\ZZ}(\beta,S,\Lambda_1) \leq \lf(2S\ell^3 + 1\ri)e^{C'\ell^{10}S^{-1}\log^2 S} \tilde{\ZZ}_0(\beta,S,\Lambda_1),
		\eeq
		where $ \tilde{\ZZ}_0(\beta,S,\Lambda_1) $ stands for the partition functions of a free Bose gas in a box $ \Lambda_1 $ with free conditions at the boundary and constraint on the total number of particles $ N_{\L_1} \leq \lf(S - S_{\star}\ri) \ell^3 $. 

The estimate of $ \tilde{\ZZ}_0 $ is very simple: we rewrite the partition function in terms of the Neumann Fourier modes $ \kv \in \latt_{1,N}^*$ (see Appendix \ref{app.tec})
and drop the constraint on the total particle number except for the mode $ \kv = {\bf 0} $ and obtain
\be
\label{step 6}
\tilde{\ZZ}_0(\beta,S,\Lambda_1) \leq \sum_{\substack{\{ \tnk \},\ \kv \in  \latt_{1,N}^* \\ 
\no \leq \lf(S - S_{\star}\ri) \ell^3}} \exp \bigg\{ - \bett J \sum_{\kv \in \latt_{1,N}^*} \eps(\kv) \tnk \bigg\}= 
\lf(S - S_{\star} \ri) \ell^3 \prod_{\substack{\kv \in \Lambda_{1,N}^*\\ \kv \neq 0}} \frac{1}{1 - \exp\big\{ - \bett \eps(\kv) \big\}}.\ee
Putting together all the estimates, we obtain
\be\frac{f(\beta,S,\Lambda)}{S} \geq \frac1{\ell^3\tilde\b}\sum_{\kk\neq{\bf 0}}\log \lf(1-e^{-\tilde\b J\e(\kk)} \ri)
-({\rm const.})\lf[ \frac{\ell^{7} (\log S)^2}{S} +\frac{\log(\ell^3 S)}{\ell^3}\ri]\;.\eeq
Finally, replacing the Riemann sum in the r.h.s. by the corresponding integral, we get 
\be \frac{f(\beta,S,\Lambda)}{S} \geq \frac1{\tilde\b}\int_{\mathcal{B}}\frac{\diff \kk}{(2\p)^3}
\log \lf(1-e^{-\tilde\b J\e(\kk)} \ri)
-({\rm const.})\lf[ \frac{\ell^{7} (\log S)^2}{S} +\frac{\log(\ell^3 S)}{\ell^3}+\frac1{\ell}\ri]\;.\ee
Optimizing over $ \ell $ yields $ \ell = S^{1/8} (\log S)^{-1/4} $ and the desired lower bound on the 
specific free energy. 
	\end{proof}

\subsection{Upper Bound}
\label{proofs ub: sec}

In this section we complete the proof of Theorem \ref{free energy: teo} by deriving a suitable (rough) upper bound to the free energy. The most relevant steps in the proof are an energy localization into Dirichlet boxes, the introduction of a fictitious magnetic field (which will be removed at the end) to take into account the constraint on the bosonic Hilbert space and the estimate proven in Proposition \ref{est kb: pro} on the interaction. A finer upper bound, supposedly optimal up to 
corrections of the order $1/S$ included, will be presented elsewhere \cite{CGS}.  

	\begin{pro}[Free energy upper bound]
		\label{free energy ub: pro}
		\mbox{}	\\
		Assume that $ \beta = \bett S^{-1}  $ as $ S \to \infty $ for some $ \bett> 0 $ constant.	Then one has
\beq
\label{free energy ub}
\frac{f(S,\beta)}S \leq   \frac{1}{\tilde\b} \int\limits_{\mathcal{B}}\frac{\diff\kk}{(2\p)^3}\log \lf( 1 -e^{ - \bett J \eps(\kv)} \ri)  + \OO \lf(S^{-1/6} (\log S)^{2/3} \ri) .
\eeq	
		
	\end{pro}

	\begin{proof}
We start by adding an external magnetic field $ h > 0$ to the Hamiltonian: exploiting the 
positivity of $ \spin_{\xv}^z + S $, one has the trivial inequality
\be \ZZ(\beta,S,\Lambda) \geq \ZZ_h(\beta,S,\Lambda):= \tr_{\HH_{\L}} \lf( \exp \lf\{ - \beta H_{h} \ri\} \ri)\ee
for any $ h > 0 $, where $H_h=H+h\sum_{\xx\in\L}(\hat S^3_\xx+S)$. We fix $\ell\in\mathbb{N}$ and 
define the 
corridor $\mathcal C$ of width $1$ as the minimal connected set on $\L$ that contains $\{\xx\in\L: 
\xx={\bf n}\ell, {\bf n}\in\mathbb{Z}^3\}$ (here we assume for simplicity that $L$ is divisible by $\ell$).
Note that $\L\setminus \mathcal{C}$ is a union 
of boxes $\L_i$ of side length $\ell-1$. We localize the energy in the boxes $ \Lambda_i $ with Dirichlet boundary conditions, by proceeding as follows. Let us denote by $\HH_{ \mathcal{C}}$ 
and $\HH_{\L\setminus \mathcal{C}}$ the Hilbert spaces generated by the spins in $\mathcal{C}$ and $\L\setminus \mathcal{C}$, respectively. Note that, if $|\!\!\downarrow_{\mathcal{C}}\rangle$ 
is the state of $\HH_{ \mathcal{C}}$ such that $S^3_{\xv} = 
-S, \forall \xv \in \mathcal{C}$, and $|\psi\rangle$ is a generic state of $\HH_{\L\setminus \mathcal{C}}$,
then $\big(\langle \downarrow_{\mathcal{C}}\!\!|\otimes\langle\psi|\big)H_h\big(|\!\!\downarrow_{\mathcal{C}}\rangle\otimes|\psi\rangle\big)=\expval{\psi}{\psi}{\sum_iH_{h,\L_i}^D}$, where
\be H^D_{h,\L_i} =H_{\L_i}^N+h\sum_{\xx\in\L_i}(\hat S^3_\xx+S)+J
\sum_{\xx\in\dpr\L_i}(S^2+S\hat S^3_\xx)\;,\ee
and the first sum runs over pairs of nearest neighbor sites, both belonging to $\L_i$, while the second 
term is the boundary contribution, which should be thought of as an interaction term between the 
sites at the boundary of $
\L_i$ with the neighboring sites in the corridor $\mathcal{C}$. Obviously, the Hamiltonians 
$H^D_{\L_i} $ all commute among each other and can be naturally thought as operators on $\HH_{\L_i}$. We also denote by $E_i^D$ and $|E_i^D\rangle$ the eigenvalues and eigenvectors of $H^D_{\L_i}$ and set $ \ket{ \{ E_i^D \}} :=
\bigotimes_{i=1}^{L^3/\ell^3} \ket{E_i^D}$.  Given these definitions, we note that
\bea \label{ub est 1}
\mathcal{Z}_{h}(\beta,S,\Lambda) &\geq& \sum_{ \{ E_i^D \}} \big(\langle \downarrow_{\mathcal{C}}\!\!|\otimes \big\langle \{ E_i^D \} \big| \big)e^{-\b H_h}\big(|\!\!\downarrow_{\mathcal{C}}\rangle\otimes \big|
 \{ E_i^D \} \big\rangle \big)\nonumber\\
&\geq& \sum_{ \{ E_i^D \}} \exp\Big\{-\b \big(\langle \downarrow_{\mathcal{C}}\!\!|\otimes\big\langle \{ E_i^D \} \big|\big)H_h\big(|\!\!\downarrow_{\mathcal{C}}\rangle\otimes\big|
 \{ E_i^D \} \big\rangle\big)\Big\}\nonumber\\
&=& \prod_{i = 1}^{L^3/\ell^3} \sum_{ E_i^{D} }e^{-\b E_i^D}  = : 
\lf( \mathcal{Z}^{D}_{h}(\beta,S,\Lambda_1) \ri)^{L^3/\ell^3},\nonumber\eea
where, in order to go from the first to the second line, we used Jensen inequality to lift the expectation value to the exponent, while to go from the second to the third we used that the expression 
in braces is equal to $-\b \expval{\{ E_i^D \}}{\{ E_i^D \}}{\sum_iH_{\L_i}^D}$. 

Now we use the bosonic representation, drop from the  partition function 
$ \mathcal{Z}^{D}_{h}(\beta,S,\Lambda_1) $ the contribution from large total occupation number 
$ N_{\L_1}$, estimate the interaction $ \kb_{\L_1}^D $ and finally restore the missing part of the 
partition function. To this purpose we introduce a new parameter $ \bar{N} \ll S $, which is going to be 
chosen later and notice that the bosonic analogue of $ H_{h, \Lambda_1}^{D} $ is 
$ \mathcal{H}_{h,\L_1}^D=\mathcal{H}_{\L_1}^N
+h\sum_{\xx\in\L_1}n_\xx+SJ\sum_{\xx\in\dpr\L_1}n_\xx$. Now, in analogy with the case
of Neumann boundary conditions, we can rewrite $ \mathcal{H}_{h,\L_1}^D= \mathcal{H}_{0,h,\L_1}^D+ \mathcal{K}_{h,\L_1}^D$, where  $\mathcal{H}_{0,h,\L_1}^D$ stands for the quadratic part of 
$ \mathcal{H}_{h,\L_1}^D$, while $ \mathcal{K}_{h,\L_1}^D$ is at least quartic in the Bose operators
and can be bounded exactly as in Eq.(\ref{kb est}). Therefore,
\beq\label{stepp 3}
\mathcal{Z}^{D}_{h}(\beta,S,\Lambda_1) 
\geq \exp \lf\{ - \frac{C \tilde\b\bar{N}^2}{S} \ri\} \tr_{\hilb_{\Lambda_1}^{0}} \lf[ \one\Big(\sum_{\xx\in\L_1} \nx \leq \bar{N}\Big) \exp 
\lf\{ - \beta \hambd_{0,h,\L_1}\ri\} \ri]\eeq
where $\hilb_{\Lambda_1}^{0}$ is the unconstrained bosonic Hilbert space on $\L_1$, i.e. 
the Bose particle number $n_\xx$ is unbounded on $\hilb_{\Lambda_1}^{0}$:
note that the removal of the constraint is irrelevant, because $ \bar{N} \ll S $ and, therefore, 
the condition that $n_\xx\le 2S$ is automatically satisfied.  
Now we rewrite $ \one\lf(\sum_{\xx\in\L_1} \nx \leq \bar{N}\ri)=1- \one\lf(\sum_{\xx\in\L_1} \nx >\bar{N}\ri)$ and bound the contribution from the states with more than $\bar N$ particles as
\bea&&	\label{remainder 1}
\sum_{\substack{\{ \nx \},\ \xv \in \latt_1\\\sum \nx > \bar{N}}} \expvalt{\{ n_{\xv} \}}{\{ n_{\xv} \}}{
\exp \lf\{ - \beta \hambd_{0,h,\L_1} \ri\}}  \le\sum_{\substack{\{ \nx \},\ \xv \in \latt_1\\\sum \nx > \bar{N}}} \Big( \{ n_{\xv} \} \Big|
\exp \Big\{ - \tilde\beta S^{-1}h\sum_{\xx\in\L_1}\hat n_\xx \Big\} \Big| \{ n_{\xv} \}  \Big)
\nonumber	\\
&&\leq e^{-\frac12\tilde\b S^{-1}h\bar N}\prod_{\xx\in\L_1}\sum_{n_\xx\ge 0}e^{-\tilde\b S^{-1}h n_\xx}
\leq e^{-\frac14\tilde\b S^{-1}h\bar N}, \eea
provided that $h/S\ll1$ and $S^{-1} h \bar{N} \gg \ell^3 \log(S/h) $.
The estimate \eqref{stepp 3} thus becomes
\beq	\label{stepp 4}
\mathcal{Z}^{D}_{h}(\beta,S,\Lambda_1)
\geq  e^{-C\tilde\b\bar N^2/S} \lf(1- e^{-\frac14\tilde\b S^{-1}h\bar N} \ri)
\tr_{\hilb^0_{\L_1}} \lf(\exp \lf\{ - \beta \hambd_{0,h,\L_1} \ri\}  \ri)\;.
\eeq
In order to compute $\tr_{\hilb_{\L_1}} \lf(\exp \lf\{ - \beta \hambd_{0,h,\L_1} \ri\}  \ri)$ we go to Fourier space: assuming for definiteness that $\L_1=\{\xx \in \Z^3 \: : \: x_i=1,\ldots,\ell-1\}$, we rewrite 
$ H_{0,h,\L_1}^D$ by applying the Fourier transform $\tilde a^\pm_\kk=\sum_{\xx\in \L_1}\varphi_\kk(\xx)a^\pm_\xx$, where 
$\varphi_\kk(\xx)=\prod_{i=1}^3\varphi_{k_i}(x_i)$,
$\varphi_k(x)=[2/(\ell-1)]^{1/2}\sin(kx)$
and the set of momenta $\L^*_{1,D}$ is $\L^*_{1,D}=\lf\{\tx\frac\p{\ell}{\bf n}, {\bf n} \in \Z^3 \: :\: n_i=1,\ldots,\ell-1\ri\}$. After Fourier transform we get
\be  \hambd_{0,h,\L_1}=SJ\sum_{\kk\in\L^*_{1,D}} \lf(\varepsilon(\kk)+\frac{h}{SJ} \ri)
\tilde a_\kk^\dagger \tilde a_\kk\;,\ee
so that 
\be \tr_{\hilb^0_{\L_1}} \lf(\exp \lf\{ - \beta \hambd_{0,h,\L_1} \ri\}  \ri)=\prod_{\kk\in\L^*_{1,D}}
\frac1{1-e^{-\tilde\b J(\varepsilon(\kk)+h/SJ)}}\;.\ee
Putting all the estimates together gives
\bea	&&\label{stepp 5}
- \frac{1}{|\Lambda_1|} \log\mathcal{Z}^{D}_{h}(\beta,S,\Lambda_1) \le\nonumber\\
&&\le \frac1{\ell^3} \sum_{\kv \in \latt_1^*, \kv \neq 0 } \log \lf(1 - e^{ - \bett J( \eps(\kv)+\frac{h}{SJ})} \ri) +({\rm const.}) \lf[\frac{\bar N^2}{\ell^3S} - \frac{1}{\ell^3} \log
\lf(1- e^{-\frac14\tilde\b S^{-1}h\bar N} \ri) \ri]\le\\
&&\le \int_{\mathcal{B}}\frac{\diff\kk}{(2\p)^3}\log \lf(1 - e^{ - \bett J\eps(\kv)} \ri)
+({\rm const.}) \lf[\frac{\bar N^2}{\ell^3S} - \frac{1}{\ell^3} \log
\lf(1- e^{-\frac14\tilde\b S^{-1}h\bar N} \ri)+\frac{h}{S}+\frac1{\ell} \ri]
\eea
To satisfy the condition  $S^{-1} h \bar{N} \gg \ell^3 \log(S/h) $ met during the proof  we can 
pick $ h = \bar{N}^{-1} S \ell^3 (\log S)^2 $. Optimizing over $\bar N$ and $\ell$ yields
\be \bar{N}  =CS^{2/3} 
(\log S)^{-2/3}\;,\qquad 	\ell =C S^{1/6} (\log S)^{-2/3} \;,\qquad  h =C S^{5/6} (\log S)^{2/3} \ll S\;,
\ee
for a suitable constant $C$. The corresponding relative error term in the energy is $ \OO(S^{-1/6}\log^{2/3}S)$, which proves the proposition. 
	\end{proof}

\appendix
\section{Two Technical Bounds}\label{app.tec}

In this appendix we prove Propositions \ref{H lower bound: pro} and \ref{est kb: pro}.

	\begin{proof}[Proof of Proposition 2.2]
	Let us assume for definiteness that $\L_1=\{\xx \in \Z^3 \: : \: x_i=1,\ldots,\ell\}$.
	We rewrite $ H_{\L_1}^N $ by applying the Fourier transform
	$
		\label{neumann fourier} \fspinv_\kk=\sum_{\xx\in \L_1}\phi_\kk(\xx)\hat{\bf  S}_\xx,
	$
	where $\phi_\kk(\xx)=\prod_{i=1}^3\phi_{k_i}(x_i)$,
	\be
		\label{neumann eigenf} \phi_k(x)=\begin{cases} \ell^{-1/2}, & {\rm if} \ k=0, \\ (2/\ell)^{1/2}\cos(k(x-\frac12)), & {\rm if}\ k\neq 0,
\end{cases}
	\ee
	and the set of momenta $\L^*_{1,N}$ is 
	\beq
		\label{neumann kappi}
		\L^*_{1,N}=\lf\{\tx\frac\p{\ell}{\bf n}, {\bf n} \in \Z^3 \: :\: n_i=0,\ldots,\ell-1\ri\}.
	\eeq
		Then a straightforward computation shows that, if $\e(\kk)=\sum_{i=1}^3(1-\cos k_i)$,  
	\beq H^N_{\L_1} 	=- 3SJ(\ell^3-\ell^2)+ J\sum_{\kk\in\L^*_{1,N}}\e(\kk)\fspinv_\kk\cdot\fspinv_\kk.\eeq
Now, for every $\kk\neq{\bf 0}$, one has $\e(\kk)\ge c_0\ell^{-2}$, for a suitable $c_0>0$. Therefore,
\be  H^N_{\L_1}\Big|_{\HH_{S_T}}\ge- 3SJ\ell^3+ J\frac{c_0}{\ell^2}\sum_{\kk\neq{\bf 0}}\fspinv_\kk\cdot\fspinv_\kk\Big|_{\HH_{S_T}}
=- 3SJ\ell^3+ J\frac{c_0}{\ell^2} \lf[S(S+1)\ell^3-\frac{S_T}{\ell^3}(S_T+1) \ri]\;,\label{B4}\ee
where in the last inequality we used Plancherel's identity 
$\sum_{\kv\in\latt_{1,N}^*} \fspinv_{\kv}^2 = \sum_{\xv \in \latt_1} \spinv_{\xv}^2 = S(S+1)\ell^3$
and the definition of total spin operator $ \tspinv = \ell^{3/2} \fspinv_0 $, from which $\fspinv_0^2 = \ell^{-3} S_T(S_T+1)$ on $\HH_{S_T}$. Now, the expression in square brackets in the r.h.s. of 
Eq. \eqref{B4}  can be bounded from below by $S\ell^3(S-S_T/\ell^3)$, so that 
\be H^N_{\L_1}\Big|_{\HH_{S_T}}\ge - 3SJ\ell^3+ c_0JS\ell \lf(S-\frac{S_T}{\ell^3} \ri)\;,\ee
Therefore, for all $ S_T $ such that $ S - \ell^{-3}S_T  > 6c_0^{-1}\ell^2 $,	we get $ H^N_{\L_1} \Big|_{\HH_{S_T}}\geq \frac{c_0}{2}J\ell S (S - \ell^{-3}S_T ) $, which proves the proposition with $ c : = \frac12J c_0 $.
	\end{proof}

	\begin{proof}[Proof of Proposition 2.3]
		We start by investigating the part of $ \kb^N_{\L_1} $ associated with the square roots in the first term in \eqref{hambh}. Using the fact that $A_{\xx,\yy}:= 1 -  \sqrt{1 - \frac{\hnx}{2S}} \sqrt{1 - \frac{\hn_{\yv}}{2S} } $ is a non-negative operator on the bosonic
Hilbert space of interest, we get:
\bea  2SJ\Big| \sum_{\media{\xx,\yy}\subset \L_1} \upx A_{\xx,\yy} a_\yy\Big|&=&
2SJ\Big| \sum_{\media{\xx,\yy}\subset \L_1} \upx A_{\xx,\yy}^{1/2}A_{\xx,\yy}^{1/2} a_\yy\Big|
\le
2SJ \sum_{\media{\xx,\yy}\subset \L_1} \upx A_{\xx,\yy} a_\xx\le\nonumber\\
&\le&  J
 \sum_{\media{\xx,\yy}\subset \L_1} \upx (n_\xx+n_\yy) a_\xx\le ({\rm const.}) \bigg( \sum_{\xv \in \latt_1} \hnx \bigg)^2\;,\eea
where in the first line we used the Cauchy-Schwarz inequality, while to go from the first to the second 
line we used the simple fact that $1-\sqrt{1-x}\le x$, $\forall x\in[0,1]$. Finally, the 
remaining term in $ \kb^N_{\L_1} $ can be bounded trivially as $J\lf| \sum_{\media{\xx,\yy}\subset \L_1} n_\xx n_\yy\ri|\le ({\rm const.}) \lf( \sum_{\xv \in \latt_1} \hnx \ri)^2$, which implies
the desired estimate.
\end{proof}
\vskip.5truecm
\noindent
{\bf Acknowledgements.} We gratefully acknowledge  financial support from the
ERC Starting Grant CoMBoS-239694. We thank E.H. Lieb, B. Nachtergaele, R. Seiringer and J.-P. Solovej for 
useful comments and discussions.

\end{document}